\documentclass[a4paper,USenglish]{article}
\usepackage[utf8]{inputenc}
\usepackage{amsmath}
\usepackage{amssymb}
\usepackage{amsthm}
\usepackage{graphicx}
\usepackage{url}
\usepackage{hyperref}
\usepackage{cite}
 
\title{Hyperbolic triangulations and discrete random graphs}

\author{Eryk Kopczyński, Dorota Celińska-Kopczyńska}

\def\map{\mu}
\def\bbN{{\mathbb N}}
\def\bbH{{\mathbb H}}
\def\bbR{{\mathbb R}}

\def\ra{\rightarrow}

\def\dist{\delta}
\def\tally{\textsc{Tally}}
\def\edgetally{\textsc{Edgetally}}

\newtheorem{theorem}{Theorem}
\newtheorem{definition}[theorem]{Definition}

\newtheorem{proposition}[theorem]{Proposition}
\newtheorem{conjecture}[theorem]{Conjecture}
\newtheorem{intuition}[theorem]{Intuition}
\numberwithin{theorem}{section}

\begin{document}

\maketitle

\begin{abstract}
The hyperbolic random graph model (HRG) has proven useful in the analysis of scale-free networks,
which are ubiquitous in many fields, from social network analysis to biology. However, working with
this model is algorithmically and conceptually challenging because of the nature of the distances in the hyperbolic plane.
In this paper we study the algorithmic properties of regularly generated triangulations in the 
hyperbolic plane. We propose a discrete variant of the HRG model where nodes are mapped to the
vertices of such a triangulation; our algorithms allow us to work with this model in a simple yet efficient way.
We present experimental results conducted on real world networks to evaluate the practical benefits of
DHRG in comparison to the HRG model.
\end{abstract}

\vspace{2em}


\newpage

\def\gshort#1#2#3{G_{#1#2#3}}

\def\geoz{\gshort 610}
\def\ghoz{\gshort 710}
\def\ghoo{\gshort 711}
\def\gooz{\gshort 810}
\def\gqab{\gshort qab}

\section{Introduction}

Hyperbolic geometry has been discovered by the 19th century mathematicians wondering about
the nature of parallel lines. One of the properties of this geometry is that the amount
of area in distance $d$ from a given point is exponential in $d$; intuitively, 
the metric structure of the hyperbolic plane is similar to that of an infinite binary
tree, except that each vertex is additionally connected to two adjacent vertices on the same
level. 

\begin{figure}[ht]
\begin{center}
\includegraphics[width=.24\textwidth]{tiling-hep.pdf}
\includegraphics[width=.24\textwidth]{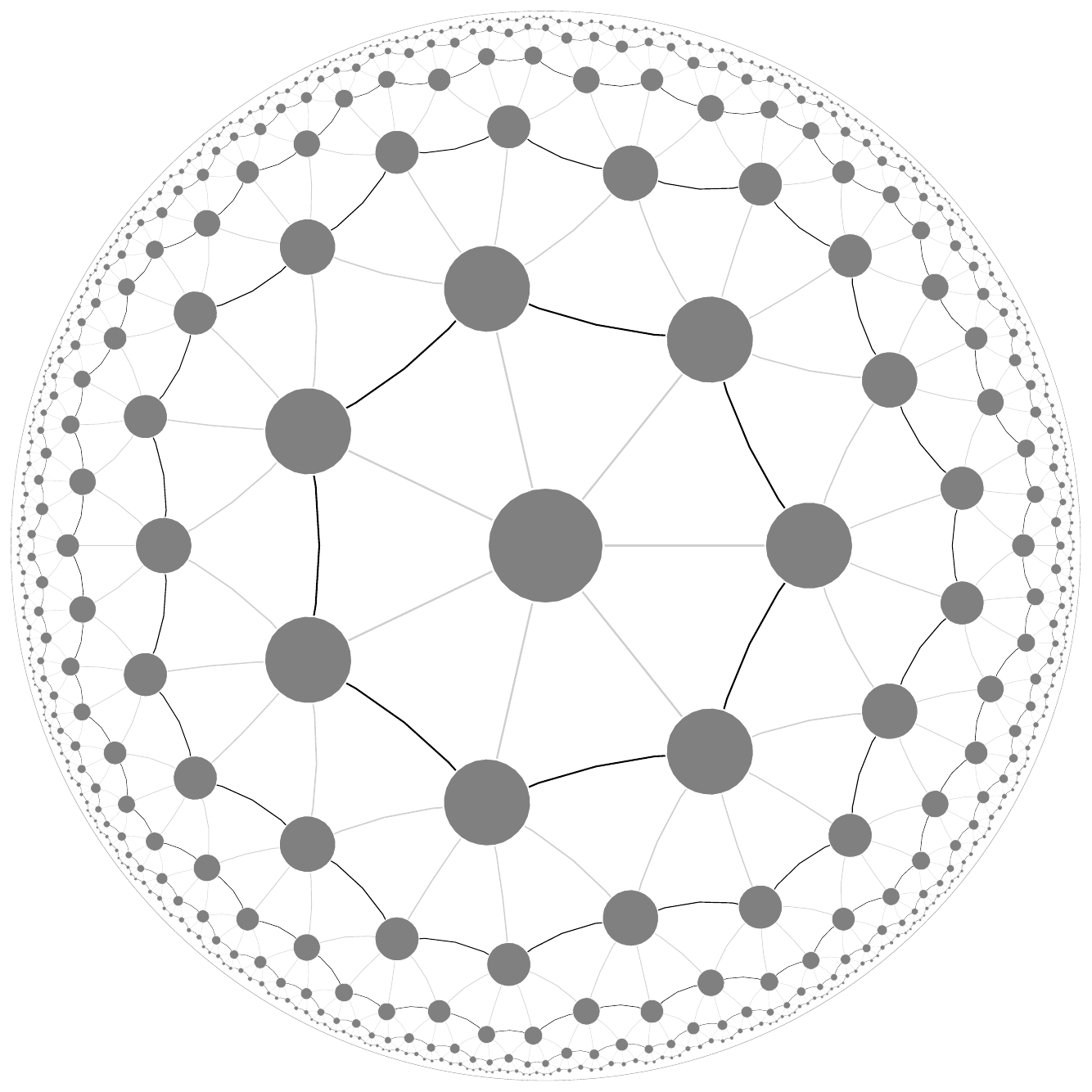}
\includegraphics[width=.24\textwidth]{tiling-hr.pdf}
\includegraphics[width=.24\textwidth]{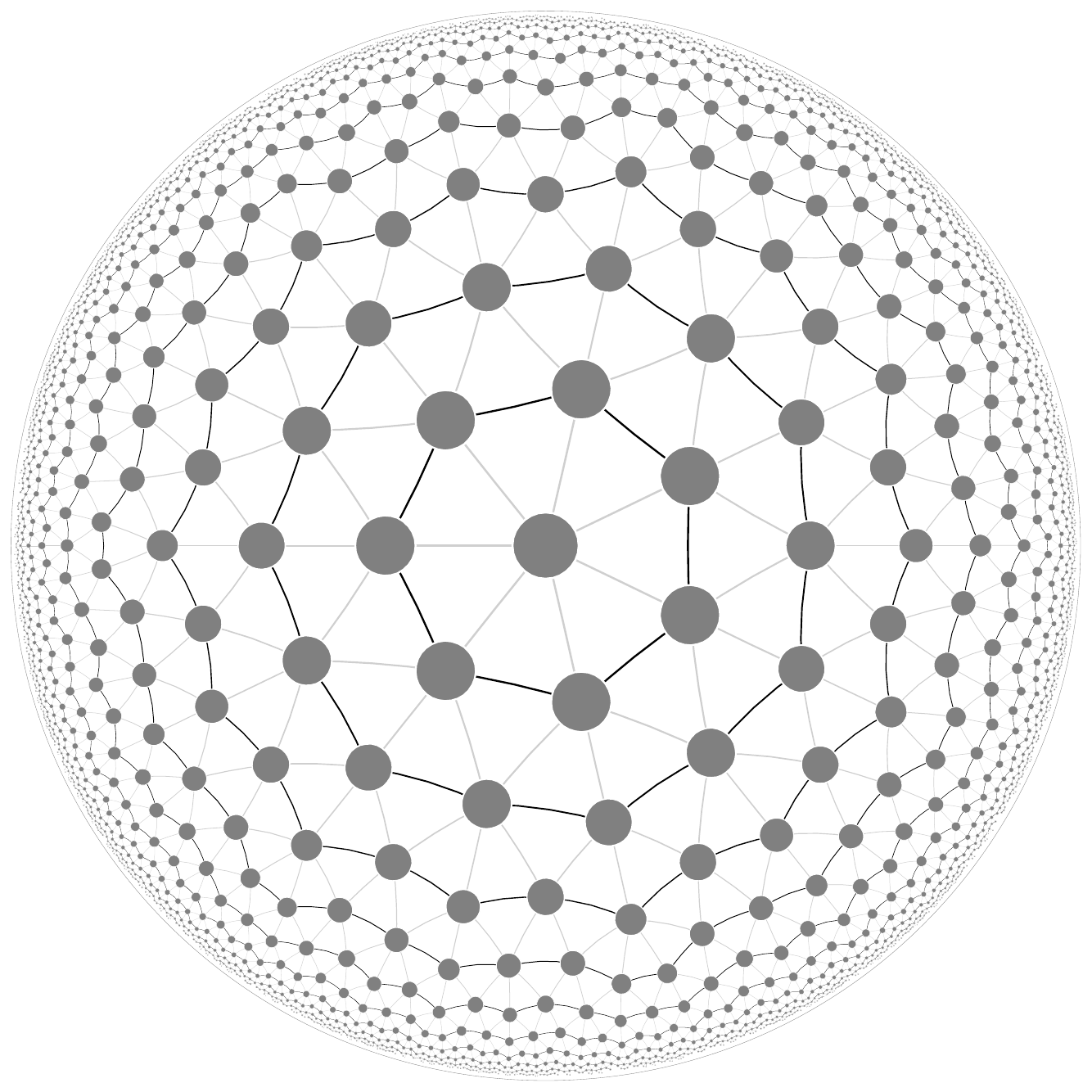}
\end{center}
\caption{\label{figtile}
(a) order-3 heptagonal tiling, (b) the triangulation $\ghoz$,
(c) truncated triangular tiling, (d) the triangulation $\ghoo$.}
\end{figure}

Figure \ref{figtile} shows two tilings of the hyperbolic plane, the order-3 heptagonal
tiling and its bitruncated variant, in the Poincar\'e disk model, together
with their dual graphs, which we call $\ghoz$ and $\ghoo$. In the Poincar\'e
model, the hyperbolic plane is represented as a disk. 
In the hyperbolic metric, 
all the triangles, heptagons and hexagons on each of these pictures are actually of the same size,
and the points on the boundary of the disk are infinitely far from the center.

Recently, hyperbolic geometry has found application in the analysis of scale-free networks,
which are ubiquitous in many fields, from network analysis to biology
\cite{papa}. 
Fix a radial coordinate system in the hyperbolic plane $\bbH^2$, where every point
is represented by two coordinates $(r,\phi)$, where $r$ is the distance from the fixed central
point, and $\phi$ is the angle from the reference direction.

\begin{definition}
\label{def_hrg}
The {\bf hyperbolic random graph model}
has four parameters: $n$ (number of vertices), $R$ (radius), $T$, and $\alpha$. Each vertex $v \in V(H) = \{1,\ldots,n\}$
is independently randomly assigned a point $\map(v) = (r_v, \phi_v)$, where
the distribution of $\phi_v$ is uniform in $[0,2\pi]$, and
the density of the distribution of $r_v \in [0,R]$ is given by 
$f(r) = { {\alpha \sinh(\alpha r)} \over {\cosh(\alpha R)-1}}$.
Then, for each pair of vertices
$v,w\in V(H)$, they are independently connected with probability $p(\dist(\map(v),\map(w)))$,
where $\dist(x,y)$ is the distance between $x,y \in \bbH^2$, and $p(d) = {1 \over 1+e^{(d-R)/2T}}$.
\end{definition}

It is known that, for correctly chosen values of $n$, $R$, $T$ and $\alpha$,
the properties of hyperbolic random graph, such as its degree distribution
or clustering coefficient, are similar to those of real world scale-free networks
\cite{gugelman}.
Perhaps the two most important algorithmic problems related to HRGs are {\it sampling}
(generate a HRG) and {\it MLE embedding}: given a real world network $G$, map the vertices of
$G$ to the hyperbolic plane in such a way that the edges are predicted as accurately as possible.
The quality of this prediction is measured with \emph{log-likelihood},
computed with the formula
$\log L(\map) = \sum_{v<w \in V(H)} \log p_{\{v,w\}\in E}(\dist(\map(v), \map(w)))$,
where $p_\phi(d)=p(d)$ if $\phi$ is true and $1-p(d)$ if $\phi$ is false.
These problems are non-trivial, as we have
to sum over all pairs of vertices (thus an $O(n^2)$ algorithm) just to compute the log-likelihood.
The original paper \cite{papa} used an $O(n^3)$ algorithm. Efficient algorithms have been found
for generating HRGs in time $O(n)$ \cite{gengraph} and for MLE embedding real world scale-free networks
into the hyperbolic plane  in time ${\tilde O}(n)$ \cite{tobias}, which was a major improvement over 
previous algorithms \cite{hypermap,vonlooz}. The algorithm in \cite{tobias}, which we call here the
BFKL embedder, is based on an $O(n)$ method
of approximating the log-likelihood.

Triangulations such as $\ghoo$ and $\ghoz$ from Figure \ref{figtile} can be 
naturally interpreted as metric spaces, where the points are the vertices of the triangulations,
and the distance $\dist(v,w)$ is the number of edges we have to traverse to reach $w$ from $v$.
Such metric spaces have properties similar to the underlying hyperbolic plane;
this similarity is much stronger than in the case of Euclidean triangulations.
In particular, 
hyperbolic shapes such as straight lines, circles, equidistant curves or horocycles have
their natural counterparts in the discrete world with very similar properties.
This similarity can be defined more formally by saying that our triangulations are
{\it Gromov hyperbolic spaces} \cite{gromovhyp}. A metric space is {\it Gromov hyperbolic} iff every 
geodesic triangle is $d$-slim, for some finite $d$. A {\it geodesic} from $u$ to $v$ is a path of length $\dist(u,v)$,
and a geodesic triangle consists of a geodesic $g_{uv}$ from $u$ to $v$, $g_{vw}$ from $v$ to $w$, and $g_{wu}$ from $w$ to $u$.
Such a triangle is $d$-slim iff every point on $g_{wu}$ lies in distance at most $d$ from $g_{uv} \cup g_{vw}$.
Since for trees $d=0$, 
Gromov hyperbolicity (i.e., the value of $d$) can be seen as a measure of tree-likeness.

{\bf Our contribution.} 
We propose a discrete analog of the HRG model, which we call the DHRG model:
in our model, $\map$ maps the nodes to the vertices of a triangulation, and 
the probability of two nodes $v_1,\ v_2$ being connected depends on the graph distance
between the vertices $\map(v_1)$ and $\map(v_2)$.

Such a discrete model lets us use a data structure we call the {\it tally counter}.
The tally counter represents a set $S$ of vertices of a triangulation; we can add
and remove vertices to it, and we can also answer queries of the form {\it for the given vertex $w$,
how many vertices in $S$ are in distance $d$ from $w$,  where $d=0, \ldots, 2R$?}. 
This data structure lets us compute the log-likelihood of a DHRG embedding 
in $O(n+m)$ queries in a straightforward
way, which is an important step in MLE embedders. Furthermore, it
lets us to dynamically remap a vertex $v$ to another location and compute the log-likelihood of the
new embeddding in $O(1+\deg(v))$ queries.

It is well known that many
algorithmic problems can be easily solved on trees; it is also well known that
many graph problems admit very efficient algorithms on graphs that are similar to trees,
where similarity is most commonly measured using the notion of {\it tree width}
\cite{treewidth}. For example, every fixed graph property definable in 
the monadic second order logic with quantification over sets of vertices and edges
($MSO_2$) 
can be checked in linear time on graphs of fixed tree width \cite{courcelle}.
A similar thing happens in our case: tree-likeness of hyperbolic tesselations lets
us to implement all the operations of the tally counter in $O(R^2)$,
while the distance between two vertices can be computed in $O(R)$. Since hyperbolic geometry exhibits exponential
growth, $R$ is typically logarithmic in $n$.



Therefore, we can easily compute the log-likelihood of a DHRG embedding 
in time $O(nR^2+mR)$, where $n$ is the number of vertices and
$m$ is the number of edges; this matches the complexity of the approximation method in the
BFKL embedder \cite{tobias} up to $R^{O(1)}$ factors.
We believe this could be used to create an efficient
MLE embedder, using discrete versions of the methods employed by that embedder; however, this
is an area of further research. For now, we used the available implementation of the BFKL embedder to
produce HRG embeddings, and transformed them to the DHRG model by 
moving every $\map(v)$ to the nearest vertex of the triangulation.
According to our experiments, despite the approximations
introduced by our discretization, our method is
much more accurate than the one used in the BFKL embedder, and it runs in comparable time.
Another benefit of our method is its dynamic remapping property, which lets us 
improve the embeddings using a local search method%
: for every vertex $v$, try to move $\map(v)$
to all its neighbors, and keep the change if it improves the log-likelihood. One
iteration of such local search can be performed in time $O(nR^2+mR)$, and the local search
stabilizes after a small number of iterations, which is a major improvement on the 
$O(n^2)$ spring embedder implemented in the BFKL embedder. Our data structures also allow
to generate DHRGs in time $O(nR^2+mR)$. While our algorithms match the best known algorithms
up to $R^{O(1)}$ factors, we believe they have a significant advantage of simplicity:
the algorithms for distance computation and the tally counter are straightforward, 
especially for theoretical computer scientists who have experience in discrete algorithmics and automata theory
\cite{margenstern_heptagrid} rather than hyperbolic geometry.
Furthermore, efficient local search might be useful on its own \cite{localsearchapp}.

It is worth to note that the major breakthrough in
\cite{gengraph} and \cite{tobias} was achieved by using geometric structures
based on partitioning hyperbolic disks into cells of the binary tiling. This is in some sense similar
to our triangulations. 
However, we believe that avoiding the continuous representations altogether 
and working with more general hyperbolic tesselations than just the binary tiling 
makes our approach more elegant. Hyperbolic triangulations have many other applications, and they are beautiful and
interesting in their
own right. 
Exponential nature of the hyperbolic geometry makes many algorithmic problems challenging 
(for large values of $R$, it is impossible to keep the whole disk of radius $R$ in the memory)
while it proves invaluable in the visualization of hierarchical data \cite{hyptree,Munzner};
mapping vertices of the visualized graph to distinct vertices of a regular triangulation
allows for aesthetically pleasant representations of graphs \cite{hrviz}.
Apart from visualizations, hyperbolic triangulations have been used to create more efficient self-organizing maps (HSOMs)
\cite{hypersom}. 
They also arise naturally when working with bounded degree planar graphs; for example,
many constructions in \cite{planarspectra} are Gromov hyperbolic graphs. 
Hyperbolic geometry is useful in mathematical art and game design \cite{hyperrogue}.
Our algorithms for computing distances in hyperbolic tesselations have found application 
in data vizualization \cite{hrviz} and in the implementation of HyperRogue \cite{hyperrogue},
which we recommend as an intuitive introduction to hyperbolic tesselations and hyperbolic geometry in general.

While the $R^{O(1)}$ factors may be seen as a disadvantage, they are
avoided in \cite{papa,gugelman,tobias} by assuming that operations on floating point numbers are
performed in time $O(1)$. However, any representation of the hyperbolic plane as a tuple of
floating point numbers in a typical 
coordinate system is prone to precision errors. Indeed, the circumference of
a hyperbolic circle of radius $r$ is $2\pi \sinh(r) = \Theta(e^r)$. Therefore, if
we are using $b$ bits for the angular coordinate, two points on the circle
of radius $b \log(2) + \Theta(1)$ will be smashed into a single point, even if
their exact distance is greater than 1. In our approach the vertices are represented
instead as paths from the ``root'' vertex, thus avoiding such precision problems even for
very large values of $R$. Even if we want to perform computations in the continuous hyperbolic plane, 
a ``hybrid'' approach where each point is represented by a vertex of our tesselation
together with the coordinates relative to that vertex is useful to prevent precision
errors. Such approach is used in HyperRogue \cite{hyperdev}.

\vskip 1em

{\bf Structure of the paper.}
In the next section we present the hyperbolic tesselations, and their properties
which will be essential for our algorithms.
Section \ref{sec:distalg} introduces our algorithms for calculating distances in the graph. 
In Section \ref{sec:distcomp}, we study how the distances in our graphs are related
to the distances in the underlying hyperbolic plane. 
We define our DHRG model in Section \ref{sec:scale-free}, based on the intuitions from Section \ref{sec:distcomp}.
We show how to apply our algorithms to
work with DHRGs efficiently in Section \ref{sec:dhrgalgo}. 
We have implemented \cite{explorable} the log-likelihood computation
and local search algorithms presented
in Sections \ref{sec:distalg} and \ref{sec:scale-free}; 
Section \ref{sec:experiments} presents
the experimental results on real world networks.
We discuss possible directions for further work in Section \ref{sec:conclusion}.
We also provide a browser-based interactive visualization of some concepts in this paper \cite{explorable}.

\section{Hyperbolic triangulations}\label{sec:htgrid}

\def\sch#1#2{\{#1,#2\}}
\def\scht#1{\sch{3}{#1}}
\def\schq#1{\sch{4}{#1}}
\def\gp#1#2{GC_{#1,#2}}

In a regular tesselation every face is a regular $p$-gon, and every vertex has degree $q$ (we assume $p,q\geq 3$). We say that such a tesselation has a 
{\bf Schl\"afli symbol} $\sch{p}{q}$. Such a tesselation exists on the sphere iff $(p-2)(q-2)<4$, plane iff $(p-2)(q-2)=4$, and hyperbolic plane iff $(p-2)(q-2) > 4$.
In this paper we are most interested in triangulations ($p=3$) of the hyperbolic plane ($q>6$).

Contrary to the Euclidean tesselations, hyperbolic tesselations cannot be scaled: on a hyperbolic plane of curvature -1, every face in a $\sch{q}{p}$ tesselation, and equivalently
the set of points closest to the given vertex in its dual $\sch{p}{q}$ tesselation, will have area $\pi(q\frac{p-2}{p}-2)$. Thus, among hyperbolic triangulations of the form $\scht{q}$, $\scht{7}$
is the finest, and they get coarser and coarser as $q$ increases.

For our applications it is useful to consider hyperbolic triangulations finer than $\scht{7}$. Such triangulations can be obtained with the {\bf Golberg-Coxeter construction}, which
adds additional vertices of degree 6. 
Consider the $\scht{6}$ triangulation of the plane, and take an equilateral triangle $X$ with one vertex in point $(0,0)$ and another vertex in the point obtained by moving $a$ steps in a
straight line, turning 60 degrees right, and moving $b$ steps more. The triangulation $\gp{a}{b} T$ is obtained from the triangulation $T$ by replacing each of its triangles with a copy
of $X$ \cite{explorable}. Regular triangulations are a special case where $a=1, b=0$. For short, we denote the triangulation $\gp{a}{b} \scht{q}$ with $\gqab$.
Figure \ref{figtile}d shows the triangulation $\ghoo$.

Let $v_0$ be a vertex in a hyperbolic triangulation $G$ of the form $\gqab$. 
We denote the set of vertices of $G$ by $V(G)$.
For $v,w \in V(G)$, let $\dist(v,w)$ be the length of the shortest path from $v$ to $w$.
Below we list the properties of our triangulations which are the most important to us.

\begin{proposition}[rings]
\label{prop_rings}
The set of vertices in distance $k$ from $v_0$ is a cycle.
\end{proposition}

We will call this cycle $k$-th {\bf ring}, $R_k(G)$. 
We assume that all the rings $R_k(G)$ are oriented clockwise around $v_0$. Thus, the
$i$-th successor of $v$, denoted $v+i$, is the vertex obtained by starting from $v$
and going $i$ vertices on the cycle. The $i$-th predecessor of $v$, denoted $v-i$,
is obtained by going $i$ vertices backwards on the cycle. A {\it segment} is the
set $S = \{v, v+1, \ldots, v+k\} \subsetneq R_k(G)$ for some $v \in V$ and $k \geq 0$;
$v$ is called the leftmost element of $S$, and $v+k$ is called the rightmost element
of $S$. By $[v,w]$ we denote the segment such that $v$ is its leftmost element, 
and $w$ is its rightmost element.
For $v,w \in R_k(G)$, let $w-v$ be the smallest $i \geq 0$ such that $w = v+i$.
We also denote $\dist_0(v) = \dist(v,v_0)$. By $B_k(G)$ we denote the $k$-th ball
(neighborhood of $v_0$), i.e.,~$B_k(G) = \bigcup_{i=0,\ldots,k} R_k(G) = \{v \in V | \dist(v,v_0) \leq k\}$.

\begin{proposition}[parents and children]
\label{prop_pch}
Every vertex (except the root $v_0$) has at most two parents and at least two children. 
\end{proposition}

We use tree-like terminology for connecting the rings. A vertex $w$ is a {\bf parent} of
$v$ if there is an edge from $v$ to $w$ 
and $\dist_0(v)=\dist_0(w)+1$; in this case, $v$ is a {\bf child}
of $w$. Let $P(v)$ be the set of parents of $v \in R_k(G)$; it forms a segment of 
$R_{k-1}(G)$, and its leftmost and rightmost elements are 
respectively called the {\bf left parent}
$p_L(v)$ and the {\bf right parent} $p_R(v)$. The set of children
$C(v)$, leftmost child $c_L(v)$ and rightmost child $c_R(v)$ are defined analogously.

\begin{figure}[ht]
\begin{center}
\includegraphics[width=.24\textwidth]{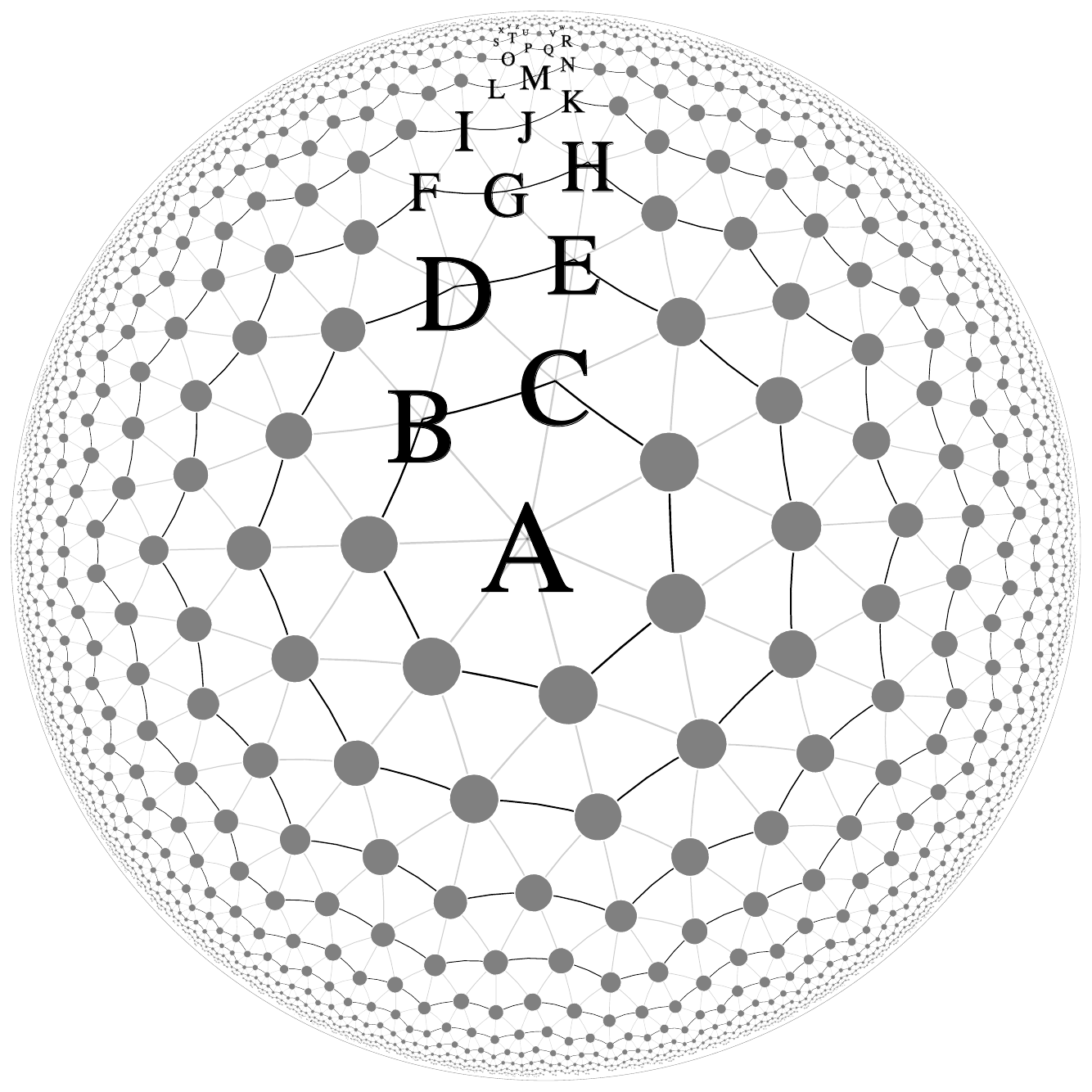}
\includegraphics[width=.24\textwidth]{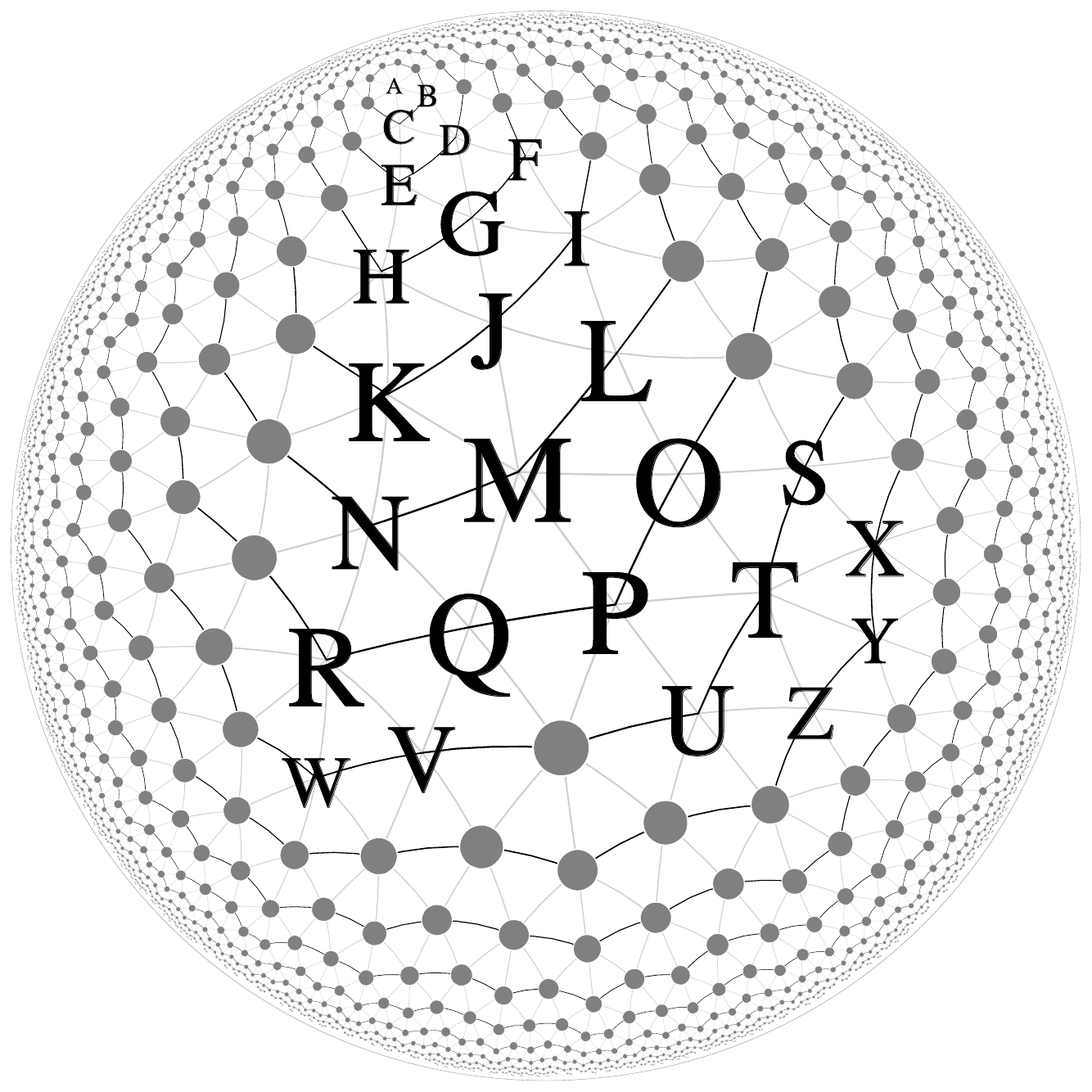}
\end{center}
\caption{\label{figlet} Triangulation $\ghoo$ with labeled vertices, in two perspectives.}
\end{figure}

Figure \ref{figlet} depicts the triangulation $\ghoo$ with named vertices. Both
pictures use the Poincar\'e disk model and show the same vertices, but the left picture is
centered roughly at $v_0$ (labeled with $A$ in the picture), and the right picture is centered at a different
location in the hyperbolic plane. Points drawn close to the boundary of the 
Poincar\'e disk are further away from each other than they appear -- for example,
vertices $T$ and $U$ appear very close in the left picture, yet in fact all
the edges are roughly of the same length (in fact, there are two lengths --
the distance between two vertices of
degree 6 is slightly different than the distance between a vertex of degree 6 and a vertex
of degree 7).

Vertices $X$, $Y$, and $Z$ are the children of $T$; its siblings are $S$ and $U$,
and its parents are $O$ and $P$. The values of $P^k([Y])$ for consecutive values
of $k$, i.e., the ancestor segments of $Y$,
 are: $[Y]$, $[T]$, $[O,P]$, $[L,M]$, $[I,K]$, $[F,H]$, $[D,E]$, $[B,C]$, $[A]$.
Vertex $W$ has just a single ancestor on each level: $R$, $N$, $K$, $H$, $E$, $C$, $A$.
Vertex $V$ has the following ancestor segments: $[Q,R]$, $[M,N]$, $[J,K]$, $[G,H]$, $[D,E]$, 
$[B,C]$, $[A]$. Note the tree-like nature of our graph: $[D,E]$ is the segment of 
ancestors for both $V$ and $Y$, and $[O,P]$ and $[Q,R]$ are already adjacent.
This tree-like nature will be useful in the algorithms in Section \ref{sec:distalg}.

\begin{proposition}[canonical shortest paths]
\label{prop_canonical}
Let $v, w \in V(G)$, and $\dist(v,w) = d$. Then 
at least one of the following is true:

\begin{itemize}
\item $v \in P^d(w)$,
\item $w \in P^d(v)$,
\item $p_R^a(v) + b = p_L^c(w)$, where $a+b+c = d$,
\item $p_R^a(w) + b = p_L^c(v)$, where $a+b+c = d$.
\end{itemize}
\end{proposition}

In other words, the shortest path between any pair of two vertices $(v,w)$ can always be
obtained by going some number of steps toward $v_0$, moving along the ring, and
going back away from $v_0$. The cases where one of the vertices is an ancestor of
the other one had to be listed separatedly because it is possible that $|P^a(v)| > 2$
for $a > 1$, thus $w$ might be neither the leftmost not the rightmost ancestor.
Such a situation happens in $\ghoo$ for the pair of vertices labeled
$(J,O)$ in Figure \ref{figlet}, even though $|P^a(v)| \leq 3$
always holds.

\begin{proposition}[regular generation]
\label{prop_rgen}
There exists a finite 
set of types $T$, a function $c: T \ra T^*$, and an assignment $t: V(G) \ra T$ 
of types to vertices, such that for each $v \in V(G)$, the sequence of types of
all children of $v$ from left to right except the rightmost child is given by $c(t(v))$.
\end{proposition}

By $T^*$ we denote the set of finite words over an alphabet $T$.
The rightmost child of $v$ is also the leftmost child of $v+1$, so we do not include
its type in $c(t(v))$ to avoid redundancy. Our function $c: T \ra T^*$ can be uniquely
extended to a homomorphism $T^* \ra T^*$, which we also denote with $c$, in the following way:
$c(t_1\ldots t_k) = c(t_1) \ldots c(t_k)$. By induction, the sequence of types of non-rightmost vertices
in $C^k(v)$ is given by $c^k(t(v))$.

For regular triangulations $\scht{q}$, the set of types is $T=\{0,1,2\}$, and
the types correspond to the number of parents \cite{explorable}. The root has type 0 and has $q$ children
of type $1$, thus $c(0)=1^q$. For a vertex with $t=1,2$ parents, the leftmost child has
type 2 (two parents), and other non-rightmost children all have type 1. Thus,
we have $c(t)=21^{q-4-t}$. Such constructions for $\sch{3}{q}$ and $\sch{4}{q}$ grids have
been previously studied by Margenstern \cite{margenstern2013small,margenstern_pentagrid,margenstern_heptagrid}.

For $\gp{1}{1}$ triangulations there are 7 types, because 
we also need to specify the degree of vertex $v$ as well as the orientation (the degree of the first child).
For Goldberg-Coxeter tesselations in general we need to identify the position of $v$ in the
triangle $X$ used in the Goldberg-Coxeter construction.

\begin{proposition}[exponential growth]
\label{prop_expgrow}
There exists a constant $\gamma(G)$ such that, for every vertex $v$,
$|C^{k}(v)| = \Theta(\gamma(G)^k)$.
\end{proposition}

Note that, if $c(t(v)) = t_1 \ldots t_n$, the number of non-rightmost vertices in $C^k(v)$ is given by 
$\sum_{i=1}^n |c^{k-1}(t_i)|$. This gives a linear recursive system of formulas for computing $|C^k(v)|$; $\gamma$
is the largest eigenvalue of the respective matrix.
We have $\gamma \approx 2.6180339$ for $\ghoz$ and $\gamma \approx 1.72208$
for $\ghoo$.

\def\tlimit{D(G)}

\begin{proposition}[Gromov hyperbolicity]
\label{prop_shortcut}
There exists a constant $\tlimit$ such that, for every $d > \tlimit$ and $x \in V(G)$, the distance from $x$ to $x+d$ is smaller than $d$.
\end{proposition}

This property gives an upper bound on the value of $b$ in Proposition \ref{prop_canonical}, and thus it
will be crucial in our algorithms computing distances between vertices of $G$.
We call this property Gromov hyperbolicity, because its combination with Proposition \ref{prop_canonical} says that the triangle with vertices in $v_0$,
$v$ and $w$ is slim. Euclidean triangulations do not have this property.


Given the canonicity of shortest paths and regular generation, the value of $D(G)$ can be found with 
a simple algorithm. We have verified experimentally for $a,b \leq 8$ that $D(\gqab) = 2a+b$.

\begin{definition}
A {\bf regularly generated hyperbolic triangulation} (RGHT) is a triangulation which satisfies all the properties
listed above.
\end{definition}

The properties above hold not only for the triangulations of the form $\gqab$. Probably the simplest, though
geometrically less regular, example of a RGHT is obtained by taking a full
infinite binary tree, and additionally connecting each vertex to its cyclic left and
right sibling, and additionally the right child of its left sibling. Such tiling has just one type 
$\star$, and $c(\star) = \star\star$.
This could be seen as a variant
of the binary tiling of the hyperbolic plane. Our algorithms will work with such tilings \cite{explorable}.

There are triangulations where the properties above do not hold; this happens even for face-transitive (Catalan) triangulations.
For example, the triangulation with face configuration V5.8.8 \cite{explorable} has vertices with three parents; this causes the tree-like distance
property to fail (consider a vertex $v$ with 3 parents and the shortest path from the leftmost parent of $v$ to $v+1$). If we
split every face of $\scht{7}$ into three isosceles triangles, we obtain the triangulation with face configuration V14.14.3 \cite{explorable}, where the sets $R_k(G)$ are no longer cycles
(vertices repeat on them), causing the regular generation to fail.
More sophisticated but qualitatively similar variants of our algorithms work for tesselations described above; we expect this to hold for any Gromov hyperbolic triangulations.
We concentrate on the regularly generated case in this paper, because non-regularly generated triangulations are much less useful for all our applications: 
they are much less uniform because of the high variance of degrees and edge lengths.

We can also consider square tilings, i.e., $\gp{a}{b} \schq{q}$ for $q \geq 5$ (Goldberg-Coxeter construction for square tilings is defined
analogously) \cite{explorable}. The major difference here is that the rings $R_k(G)$ are disconnected rather than cycles. However, this only makes our
algorithms simpler: the canonical shortest paths (Proposition \ref{prop_canonical}) no longer have to go across the ring, i.e., $b$ always
equals 0. However, despite the greater simplicity and better performance, square tilings give worse results for the HRG embedding applications. This
is not surprising, as they provide a less accurate approximation of hyperbolic distance.

Our ring structure has a singularity in $v_0$. It is possible to avoid this singularity by changing our construction a bit, by making $R_d(G)$ 
into infinite paths (horocycles) \cite{explorable}. Another possible change to our construction is to connect the last element of $R_d(G)$ with the adjacent element
of $R_{d+1}(G)$, thus putting all the vertices of $G$ in a single spiral \cite{planarspectra}.

\section{Computing distances in hyperbolic triangulations}\label{sec:distalg}

It is not feasible to represent all vertices in, say,
$B_{100}(\ghoo)$ in computer memory -- there are more than $10^{23}$ of them!
However, Proposition \ref{prop_rgen} lets us generate the vertices in our RGHT
lazily. That is, represent our vertices with pointers, start from the root, and
generate other vertices when asked for them. In particular, 
each vertex $v$ is represented with
a pointer to a structure which contains $\dist_0(v)$, the type of $v$, the pointers to
$p_L(v)$, $p_R(v)$, $v-1,$ $v+1,$ $c_L(v),$ and the index of $v$ among the children
of $p_R(v)$; the last three pointers are NULL if the
given neighbor has not yet been computed. Such a structure allows us to compute
all the neighbors of the given vertex in amortized time $O(1)$ for a fixed triangulation.
In this section we show how to compute distances in a RGHT, based on
this data structure.


\begin{theorem}\label{distalgo}
Fix a RGHT $G$. Then $\dist(v,w)$ can be computed 
for $v,w \in G$ in time $O(\dist(v,w))$.
\end{theorem}

\begin{proof}[Proof (sketch)]
The idea of the algorithm is to find the shortest path given in Proposition
\ref{prop_canonical} and limited according to Proposition \ref{prop_shortcut}.
Suppose that $\dist_0(v) = d' + \dist_0(w)$, where $d' \geq 0$.
For each
$i$ starting from 0 we compute the endpoints of the segments $P^{d'+i}(v)$ and
$P^i(w)$. We check whether these segments are in distance at most $\tlimit$ on the ring;
if no, then we can surely tell that we need to check the next $i$; if yes,
we know that the shortest path can be found on one of the levels from $i$ to 
$i+\lfloor \tlimit/2 \rfloor$. We compute the length of all such paths and return the
minimum. 
The full algorithm and the proof of its correctness is given in the Appendix \ref{ommited}.
\end{proof}

It is worth to note that $D(\gshort k11) = 3$ and $D(\gshort k10) = 2$; these RGHTs are most appropriate for our
applications, and our algorithm is very efficient for them.
With some preprocessing, we can optimize to $O(\log \dist(v,w))$ per query
-- precompute $p_L^a(v)$  for each $v \in V$ and $a$ that
is a power of two.

A {\bf distance tally counter} for a graph $G=(V,E)$ represents a modifiable function $f: V \ra \bbR$
with the following operations:

\begin{itemize}
\item Initialize: $f$ is initialized with the constant 0 function
\item Add($v$, $k$): add $k$ to $f(v)$
\item Tally($v$): return an array $A$ such that, for every $d \in \bbN$, 
$A[d] = \sum_{w \in V: \dist(v,w)=d} f(w)$ (if $d$ is out of bounds of $A$, 
we assume that $A[d]=0$)
\end{itemize}

\begin{theorem}\label{powerdistalgo}
Fix a RGHT $G$. A distance tally counter can be implemented working
in memory $O(\sum_{w \in W: f(w) \neq 0} \delta_0(w)^2)$, initialization in time $O(1)$,
and Add($v$) and Tally($v$) in time $O(\dist_0(v)^2)$.
\end{theorem}

\begin{proof}[Proof (sketch)]
A segment is {\bf good} if it is of the form $P^d([v,v])$ for some $v \in V$ and
$d \in \bbN$. Note that the algorithm from the proof of Theorem \ref{distalgo} can be seen as follows:
we start with two segments $[v,v]$ and $[w,w]$, and
then apply the operation $P$ to each of them until we obtain good segments which
are close. Our algorithm will optimize this by representing all the good segments
coming from vertices $v$ added to our structure.

We call a vertex or good segment is called {\it active} if it has been already generated,
and thus is represented as an object in memory.
For each active vertex $v \in V$ we keep two lists $L_L(v), L_R(v)$
of active segments $S$ such
that $v$ is respectively the leftmost and rightmost element of $S$. 
Each active segment $S$ also has a pointer to
$P(S)$, which is also active (and thus, all the ancestors of $S$ are active too), and a
dynamic array of integers $a(S)$. Initially, there are no active vertices or good segments;
when we activate a segment $S$, its $a(S)$ is initially filled with zeros.

The operation Add($v$, $k$) activates $v$, and $S = [v,v]$ together with all its
ancestors. Then, for each $i = 0, \ldots, \dist_0(v)$, it adds $k$ to $a(P^i(S))[i]$.

The operation Tally($w$) activates $w$ and $S = [w,w]$ together with all its ancestors.
We return the vector $A$ obtained as follows.
We look at $p^i(S)$ for $i = 0, \ldots, \dist_0(v)$, and for each $p^i(S)$,
we look at close good segments $q'$ on the same level, baswed on the lists $L_L(w), L_r(w)$
for all $w$ in distance at most $\tlimit$ from $p_i(S)$.
The intuition here is as
follows: the algorithm from Theorem \ref{distalgo}, on reaching $p^{i_1}(v)=S$ and $p^{i_2}(w)=S'$, 
would find out that these two pairs are close enough and return $i_1+i_2+\dist(S,S')$; in our
case, for each $c$ such that $a(S')[c] \neq 0$, we will instead add $a(S')[c]$ to
$A[a_1+\dist(S,S')+c]$. We have to make sure that we do not count
vertices which have been already counted.
\end{proof}


\section{Graph distances versus hyperbolic distances}\label{sec:distcomp}

Let $j: V(G) \ra \bbH^2$ be the function mapping the vertices of our triangulation to their
position on the hyperbolic plane. $j(v)$ can be computed by applying $d=\dist_0(v)$ isometries to $j(v_0)$, with $i$-th
isometry depending only on the type of $p_R^{d-i+1}(v)$ and the index of
$p_R^{d-i}(v)$ among its children.

\begin{intuition}\label{distint}
For $v, w \in V(G)$, let $d = \dist(v,w)$, and $r=\dist(j(v), j(w))$.
Then $d$ and $r$ are approximately proportional. 
\end{intuition}

Stating and proving this intuition formally appears to be challenging, as we have to deal both with
the discrete structure of the triangulation, and the continuous hyperbolic geometry.
From the regularity of our tesselation we get
that $d=\Theta(r)$; we cannot
give a better estimate (e.g., $d=\alpha r+\Theta(1)$)
because the density of rings depends on the direction.
However, we can guess that, on average, $r \approx d \log \gamma$.
This is because, in the hyperbolic plane, the area and circumference
of a circle of radius $r$ given in absolute units is given by $\cosh(r)-1$ and
$\sinh(r)$ respectively, which are $\Theta(e^r)$; from Proposition 
\ref{prop_expgrow} we know that this corresponds to $\Theta(\gamma^d)$ vertices of our graph,
yielding $r \approx d \log \gamma$ after taking the logarithm of both sides.

We can also expect the grid approximation to be better than the corresponding Euclidean one.
Consider the regular triangulation $\geoz$ on the Euclidean plane, in the standard
embedding where every edge has length 1. Let $v=v_0$ and
$W$ be a random vertex in $R_d(\geoz)$. From basic geometry we obtain that
$r \in [\frac{\sqrt{3}}{2}d, d]$. The standard deviation of $r$ will be linear in
$d$, because the ratio $r/d$ depends on the angle between the line $(v_0,W)$ and the
grid lines. However, in the hyperbolic plane, because of the exponential expansion, 
this angle constantly changes as the line $(v_0,W)$ traverses the grid, leading to the
following conjecture:

\begin{conjecture}\label{distconj}
Let $G=\gqab$, and $W \in R_d(G)$ be randomly chosen. Then 
$\dist(j(v_0),j(W)) = c_1 d  + c_0 + X$, where $EX = o(1)$, $\mbox{Var}\ X = \Theta(d)$.
\end{conjecture}

The results of experimental verification agree with the conjecture for
$\ghoo$, $\ghoz$ and $\gooz$, although $c_1$ is slightly larger than $\log \gamma$
in these cases. While Conjecture \ref{distconj} remains unproven, it is worth to remind that
it is not essential to our work -- our triangulations interpreted as abstract metric spaces exhibit
hyperbolic properties in their own right.


\section{Discrete hyperbolic random graphs}\label{sec:scale-free}

In this section we use our intuitions from the previous section to define the discrete
hyperbolic random graph model (DHRG), the discrete version of the HRG model (Definition \ref{def_hrg}). 

In our model, we map vertices $v \in V(H)$ not to points in the continuous
hyperbolic plane, but to the vertices of our RGHT $G$, 
i.e., $\map:V(H) \ra V(G)$. 
The density function $f(r)$ from the HRG model cannot be reproduced exactly,
but we can use $f(r) = {\alpha e^{\alpha r}}/(e^{\alpha r}-1)$, which is a very good approximation
(it only slightly changes the low probability of placing a vertex very close to the center).

\begin{definition} A {\bf discrete hyperbolic random graph (DHRG)} over the RGHT $G$ with parameters
$n,$ $R,$ $T,$ and $\alpha$ is a random graph $H=(V(H),E(H))$ constructed as follows:
\begin{itemize}
\item The set of vertices is $V(H) = \{1, \ldots, n\}$,
\item Every vertex $v \in V(H)$ is independently randomly assigned a vertex $\map(v) \in B_R(G)$ 
in such a way that the probability that $\map(v)=w$ is proportional to $\frac{e^{d\alpha}}{|R_d(G)|}$,
where $d = \dist_0(w)$;
\item Every pair of vertices $v_1, v_2 \in V(H)$ are independently connected with an edge
with probability $p(\dist(\map(v_1),\map(v_2)))$, where $p(d) = {1 \over 1+e^{(d-R)/2T}}$.
\end{itemize}
\end{definition}

Note that the definition permits $\map(v_1)=\map(v_2)$ for two different vertices $v_1,v_2 \in V(H)$
-- this is not a problem, furthermore, such vertices $v_1$ and $v_2$ are not necessarily connected,
nor do they need to have equal sets of neighbors.

DHRG mappings can be converted to HRG by composing
$\map$ with $j$, and the other conversion can be done by finding the nearest tesselation vertex to $\map(v)$ for each $v\in V(G)$. 
From Conjecture \ref{distconj} we expect the DHRG parameters $\alpha$, 
$R$, and $T$ to be related to the HRG parameters by the factor of
$\log \gamma$.

\begin{theorem}\label{dhrg_powerlaw}
DHRG with parameters $\alpha>\log \gamma/2$, $R$, $T$ and $n$ has a power law degree distribution
with exponent $\beta = 1 + 2(\alpha / \log \gamma)$.
Furthermore, the expected clustering coefficient, average degree, and approximate degree distribution
of a DHRG with given parameters can be computed in time polynomial in $R$. (Proof in the Appendix.)
\end{theorem}

                          



\section{Algorithms for DHRG}\label{sec:dhrgalgo}

We show how the algorithms from Section \ref{sec:distalg} allow us to deal with the DHRG
model efficiently.

{\bf Computing the likelihood.} 
Computing the log-likelihood in the continuous model is difficult, because
we need to compute the sum over $O(n^2)$ pairs; a better algorithm
was crucial for efficient embedding of large real world scale-free networks \cite{tobias}.
The algorithms from the 
previous section allow us to compute it quite easily in the DHRG model.
To compute the log-likelihood of our embedding of a network $H$ with $n$ vertices and
$m$ edges, such that $\dist_0(v) \leq R$ for each $v \in V(H)$, we:

\begin{itemize}
\item for each $d$, compute $\tally[d]$, which is the number of pairs $(v,w)$ such that $\dist(v,w) = d$
-- the distance tally counter allows doing this in a straightforward way
(simply by doing Add($\map(v)$, 1) for each $v \in V(H)$), in time
$O(n R^2)$.

\item for each $d$, compute $\edgetally[d]$, which is the number of pairs $(v,w)$ connected by an edge such
that $\dist(v,w) = d$ -- this can be done in time $O(mR)$ simply by using the
distance algorithm for each of $m$ edges.
\end{itemize}

After computing these two values for each $d$, computing the log-likelihood is
straightforward. One of the advantages over \cite{tobias} is that we can then
easily compute the log-likelihood obtained from other values of $R$ and $T$,
or from a function $p(d)$ which is not necessarily logistic.

{\bf Improving the embedding.} A continuous embedding can be improved by a {\it
spring embedder} \cite{springembedder}. Imagine that there are attractive forces
between connected pairs of vertices, and repulsive forces between unconnected pairs.
The embedding $m$ will change in time as the forces push the vertices towards locations
in such a way that the quality of the embedding, measured by log-likelihood, is improved.
However, computationally, spring embedders are very expensive -- there are $O(n^2)$
forces, and potentially, many steps of our simulation could be necessary.

On the other hand, our algorithms allow to improve DHRG embeddings quite easily.
We use a local search algorithm.
Suppose we have computed the log-likelihood, and on
the way we have computed the vectors $\tally$ and $\edgetally$, as well as the distance
tally counter where every $\map(v)$ has been added. Now, let $v' \in V(H)$ be a vertex
of our embedding, and $w \in V(G)$. Let $\map'$ be the new embedding given by
$\map'(v')=w$ and $\map'(v)=\map(v)$ for $v \neq v'$. Our auxiliary data allows us then to
compute the log-likelihood of $\map'$ in time $O(R^2 + R \deg(w))$. 

This allows us to try to improve the embedding in the following way: in each step,
for each
$v \in V(G)$, consider all neighbors of $\map(v)$, compute the log-likelihood for all
of them, and if for some $\map'$ we have $\log L(\map') > \log L(\map)$, 
replace $\map$ with $\map'$. Assuming the bounded degree of $G$, this can be done in time
$O(R^2n + Rm)$.


{\bf Generating a random graph.} 
Generating large HRGs is not trivial -- a naive algorithm works in $\Theta(n^2)$;
algorithms working in $O((n^{3/2}+m) \log n)$ and $O(n)$ \cite{vonlooz,gengraph} are known.
Our algorithms allow to generate DHRGs quite easily in $O(n R^2+mR)$.

The first step is to generate the vertices. For each vertex $v=\{1, \ldots, n\}$,
we choose $d=\dist_0(v)$ (according to the given distribution), and
then we have to randomly choose $v$ from the $|R_d(v)|$ possibilities. This can
be done iteratively: we create a sequence of vertices $v_0$, \ldots, $v_d$, where
$v_0$ is the root, and $v_{k+1}$ is a non-rightmost child of $v_k$. The probability
of choosing the particular $v$ as $v_k$ should be proportional to $a_{d-k}(v)$,
where $a_i(v) = |C^{d-k}(v)|-1 = | c^{d-k} (t(v)) |$ can be obtained by matrix
multiplication ($O(R)$ preprocessing).

The second step is to generate the edges. This can be done by modifying the 
algorithm computing the vector $\tally[d]$ -- when we add $k$ to $\tally[d]$,
we now also add each of the edges with the probability $p(d)$. Thus, we
need to choose a subset of $S = \{1, \ldots, k\}$ where each element is independently
chosen with probability $p$. $\min S$ has a geometric distribution Geo($p$),
except the cases where $S=\emptyset$ which are represented by Geo($p$)>$k$; assuming
that Geo($p$) can be sampled in $O(1)$, this allows us to generate $\min S$ in time 
$O(1)$, and the rest of $S$ can then be generated in the same way. Then,
trace the elements of $S$ back to their
original vertices, which can be done in $O(R)$ per edge by following the tree of 
active segments back. The whole algorithm works in time $O(n R^2+mR)$, where $n$
is the number of vertices and $m$ is the number of generated edges.

\section{Experimental results}\label{sec:experiments}

We have implemented the log-likelihood and local search algorithms outlined in the
previous section, and conducted experiments on real world network data.
More details are in the Appendix, and the results are included with our implementation \cite{explorable}.

{\bf Facebook social circle network.}
First, we test our model on a relatively small network. We have chosen the Facebook social 
circle network, coming from the SNAP database \cite{snapnets}
and included
with the hyperbolic embedder implementing the algorithm by Bl\"asius et al \cite{tobias},
which we will refer to as BFKL.
This network has $N=4039$ nodes and $M=88234$ edges. 
BFKL has mapped this graph to the hyperbolic plane, using parameters
$R=12.576$, $\alpha=0.755$, $T=0.1$. We have computed the log-likelihood as
$L_1=-516534$. This looks extremely bad at first, as it is worse than the log-likelihood
of the trivial model where each edge exists with probability $M/{N \choose 2}$,
which is $L_0=-487133$; however, this is because the influence of the parameter $T$
on the quality of the embedding is small \cite{hypermap}, and thus BFKL uses a small value of $T=0.1$,
which does not necessarily correspond to the network. The best log-likelihood
of $L_2=-176132$ is obtained for $R_2=11.09358$ and $T_2=0.54336$.

Now, we convert this embedding into the DHRG model, by finding the nearest vertex 
of $\ghoo$ for each $v \in V(H)$. The best log-likelihood $L_3=-179125$ is obtained for 
$R_3=20.39395$ and $T_3=1.01295$;
as predicted in Section \ref{sec:scale-free}, $T_2/T_3 \approx R_2/R_3 \approx \log \gamma$. 
Our log-likelihood
$L_3$ is slightly worse than $L_2$, but this is not surprising -- first, our edge
predictor has lost some precision in the input because of the discrete nature of
our tesselation, and second, the original prediction was based on the hyperbolic distance $r$ while our prediction
is based on the tesselation distance $d$, and the ratio of $r$ and $d$ depends on the direction. We also
compute the log-likelihood obtained by a model where the edge probability is
$p(d) = \edgetally[d]/\tally[d]$, which corresponds to using the best possible 
function $p(d)$ (not necessarily logistic); we obtain $L_4=-177033$, which is only slightly
better than $L_3$. This shows that the logistic function is close to the optimum.

Now, we try our local search algorithm. The points stopped moving in the $k$-th iteration, for $k=22$.
This allows us to improve the log-likelihood of $L_5=-167991$, again for the best values
of $R_5=20.710576$ and $T_5=0.964954$, and the optimal log-likelihood to $L_6=-165338$. 

Now, we convert our mapping back to the HRG model, obtaining the log-likelihood of
$L_7=-168445$ for the optimal values of $R_7=11.17756$ and $T_7=0.52578$. Note that $L_7$ is
significantly better than $L_2$; hence, despite converting from HRG to DHRG and back,
our method was successful at finding a better continuous embedding.

The running time
of parts of our algorithm were: $t_1$=0.4 s (converting),
$t_2$=0.067 s (computing $\edgetally$), $t_3$=0.031 s (computing $\tally$),
$t_4$=40 s (local search). The BFKL embedder computes the log-likelihood in 0.3 seconds, which
is comparable. However, their spring embedder working in quadratic time is much slower
than our local search.\footnote{For $T=0.54336$ and seed 123456789
the BFKL spring embedder reported the log-likelihood of -131634,
which is better than ours; however, our implementation reports $L_1=-211454$
and $L_2=-174465$, which our local search still manages to improve to $L_7=-157026$.
This appears to to be a problem in their approximation
(which also affects the fast embedder, and smaller values of $T$).
Indeed, replacing their optimized log-likelihood function with a $\Theta(n^2)$ 
one from {\tt hyperbolic.cpp} reports log-likelihood equal to ours.
[Actually, it reports double our result, but this seems to be caused by counting
each pair of vertices twice, which is easy to fix and irrelevant for the optimized embedder.]
}


The respective values obtained on $\ghoz$ were: 
$t_1=0.5s$, $t_2=34 ms$, $t_3=19ms$, $k=29$, $t_4=22s$,
$L_3=-182721$, $L_4=-188134$, $L_5=-170074$, $L_6=-168006$, 
$L_7=-170886$. $\ghoz$ is coarser than
$\ghoo$, hence it is not surprising that its results are slightly worse; 
also the smaller size and greater simplicity of $\ghoz$ improves the running time.
Yet, the general qualitative effects are similar. Using finer triangulations such as
$\gshort 753$ yields minor improvements in the resulting log-likelihood at the cost
of significant performance downgrade, due to the increase in the values of $R$ and
$\tlimit$.


{\bf GitHub following graph.}
To benchmark our algorithm on a large network,
we study the embedding of a social network observed in GitHub repository hosting service. 
In GitHub convention, \emph{following} means that a registered user
agreed to be sent notifications about other user's activity within the service.
This relationship can be represented by the means of the graph of following $\mathcal{G}_f$.
There is an edge in $\mathcal{G}_f$ between A and B iff A follows B.
Decision about following a particular user can be simultaneously driven by their
popularity within the network and the similarity to the interested user, which suggests
hyperbolic geometry can be intrinsic in the development of $\mathcal{G}_f$. $\mathcal{G}_f$ 
was also proved to show power-law-like scale behavior \cite{euromed}, that is why we
believe it is a sound benchmark for our analysis. Since the complete download of GitHub
data is impossible, our dataset is combined from two sources: GHTorrent project
\cite{ght} and GitHubArchive project \cite{gha}. The analyzed network contains 
information about the following relationships that occurred in the service from 2008 to 
2009.

The graph has $n$=74946 vertices and $m$=537952 edges
 (since we are working with
an undirected graph, an edge appears between A and B if either A follows B or B follows
A). The BFKL embedder has chosen
parameters $R=20.9037$ and $\alpha=0.855$, and computes the log-likelihood in
5 seconds.
The results for $\ghoz$ are as follows: $t_1=12$ s, $t_2=2$ s, $t_3=0.5s$,
$L_0=-4364526$, $L_3=-3976515$, $T_3=1.398666$, $R_3=9.063012$,
$L_4=-3859688$. After 6 iterations of local search (25s each) the results have been improved to
$L_5=-3571941$, $L_6=-3542740$; after 100 iterations the results are only slightly better,
at -3545664 and -3527397. 
The time $t_2+t_3$ is still comparable
to BFKL.\footnote{As with the smaller graph, we suspect that our value is more
accurate than BFKL.} Using an even coarser 
$\gooz$ reduces the running time per iteration by about $\frac{1}{4}$, without 
a significant reduction in quality.




\section{Conclusion}\label{sec:conclusion}

We have shown efficient algorithms for computing the distances between points in
regularly generated hyperbolic triangulations, and distances between a given point and a set of points.
We have shown how to apply these algorithms to work with the DHRG model efficiently, 
and how our DHRG model can be used to improve the results of the BFKL embedder.
Creating a DHRG embedder is an direction of further research; we believe that
the ideas underlying the BFKL embedder could be applied to the DHRG case.
It is also interesting to what extent our algorithms for RGHTs can be generalized to wider classes of hyperbolic graphs, such
as graphs with Gromov hyperbolicity $\delta$ \cite{gromovhyp}.



We are very grateful to the anonymous referees for their careful reading 
of an earlier version of this work.
Many parts of the paper have been greatly improved as a result of their insightful and constructive
comments. 

\bibliography{mybib}

\newpage
\appendix

\section{Omitted proofs}\label{ommited}

\begin{figure}[ht!]
\begin{enumerate}
\def\i{\item}
\def\iz{\item}
\def\cindent{\hskip 0.5cm}

\iz {\bf function} distance$(v_1, v_2)$:
\i  \cindent {\bf for} $i \in \{1,2\}$:
\i  \cindent \cindent $l_i := v_i$
\i  \cindent \cindent $r_i := v_i$
\i  \cindent \cindent $d_i := \dist_0(v_i)$
\i  \cindent \cindent $a_i := 0$
\i  \cindent {\bf function} push($i$):
\i  \cindent \cindent $a_i := a_i+1$
\i  \cindent \cindent $d_i := d_i-1$
\i  \cindent \cindent $l_i := p_L(l_i)$
\i  \cindent \cindent $r_i := p_R(r_i)$
\i  \cindent {\bf while} $d_1 > d_2:$
\i  \cindent \cindent push(1)
\i  \cindent {\bf while} $d_2 > d_1:$
\i  \cindent \cindent push(2)
\i  \label{easystart} \cindent {\bf for} $i \in \{1,2\}$ {\bf if} $v_i \in [l_i, r_i]:$
\i  \label{easyend} \cindent \cindent \cindent return $a_{3-i}$
\i  \cindent $d := \infty$
\i  \label{mainloopstart} \cindent {\bf while} $a_1 + a_2 < d$:
\i  \cindent \cindent {\bf for} $i \in \{1,2\}$ {\bf for} $k \in \{0,\ldots,t_2\}$ {\bf if} $l_i = r_{3-i} + k:$
\i  \cindent \cindent \cindent $d := \min(d, a_1 + a_2 + k)$
\i  \cindent \cindent push(1)
\i  \label{mainloopend} \cindent \cindent push(2)
\i  \cindent {\bf return} $d$
\end{enumerate}
\caption{\label{distalgo_psc}
Pseudocode of the algorithm from Theorem \ref{distalgo}.}
\end{figure}

\begin{proof}[Proof of Proposition \ref{prop_canonical}]
Let $v, w \in V(G)$ for a triangulation $G$ satisfying the previous properties.
Let $(v=v_0, v_1, v_2, \ldots, v_d=w)$ be a path from $v_0=v$ to
$v_d=w$ of length $d$. We will show that a path from $v$ to $w$
exists which is of the form given in Proposition \ref{prop_canonical}
and is not longer than $d$.

In case if $v \in P^d(w)$ or $w \in P^d(v)$, the hypothesis trivially
holds, so assume this is not the case.

\def\ras#1{\stackrel{#1}{\ra}}

\def\ca{{v_i}}
\def\cb{{v_{i+1}}}
\def\cc{{v_{i+2}}}

Each edge from $\ca$ to $\cb$ on the path is one of the following types:
right parent, left parent, right sibling, left sibling,
right child (inverse of left parent, i.e., any non-leftmost child), left child
(inverse of right parent, i.e., any non-rightmost child).
We denote the cases as respectively $\ca \ras{RP} \cb$,
$\ca \ras{LP} \cb$, $\ca \ras{RS} \cb$, $\ca \ras{LS} \cb$, 
$\ca \ras{RC} \cb$, $\ca \ras{LC} \cb$. 
We use the symbols $x,y$
if we do not care about the sides.

If $\ca \ras{xC} \cb \ras{yP} \cc$, then we can make the path shorter
($\ca$ and $\cc$ are both children of $\cb$ and thus they must be the
same or adjacent).

If $\ca \ras{xS} \cb \ras{yP} \cc$, then let $u$ be such that $\ca \ras{yP} u$.
Either $u=\cc$ or $u$ is adjacent to $\cc$, so we can replace this situation
with $\ca \ras{yP} \cc$ or $\ca \ras{yP} u \ras{zS} \cc$, without making the
path longer. The case $\ca \ras{yC} \cb \ras{xS} \cc$ is symmetric.

Therefore, all the $xP$ edges must be before all the $xS$ edges, which must
be before all the $xC$ edges. Furthermore, clearly all the $xS$ edges must
go in the same direction -- two adjacent edges moving in opposite directions
cancel each other. 

We will now show that all the edges have to go in the same direction (right 
or left). This direction will be called $m \in \{L,R\}$. There are three cases:
\begin{itemize}
\item there are $xS$ edges -- if they do not all go in the same direction,
then two adjacent ones moving in the opposite directions cancel each other,
so we can get a shorter path by removing them. Otherwise, let $m$ be the
common direction.
\item there are no $xS$ edges, and the vertex between $xP$ edges and $xC$
edges is the root -- in this
case, we get from $v$ to the root using $a$ parent edges, and then from 
the root to $w$ using $c$ child edges. If we replace the first $a$ edges
with right parent edges, we still get to $v_0$; symmetrically, we replace
the last $c$ edges with right child edges.
\item there are no $xS$ edges, and the vertex between $xP$ edges and $xC$
edges is $v_i$ which is not the root -- then, the main direction is $R$
iff $v_{i-1}$ is to the left from $v_{i+1}$ among the children of $v_i$,
and $L$ otherwise.
\end{itemize}

Now, we can assume that all the edges in the $xC$ go in the same
direction (i.e., they are $mC$
edges). Indeed, if this is not the case, let $m'$ be the opposite of $m$,
and take the last $m'C$ edge: $\ca \ras{m'C} \cb \ras{m?} \cc$. In all
cases, let $u$ be such that $\ca \ras{mC} u$. By case by case analysis,
we get that $\ca \ra u \ra \cc$ the path is either shorter (i.e., $u=\cc$)
or pushes the $m'C$ edge further the path. Ultimately, we get no $m'C$
edges in the $xC$ part. By symmetry, we also have no $m'P$ edges in the
$mP$ part.

Therefore, our path consists of $a$ $mC$ edges, followed by $b$ $mS$ edges,
followed by $c$ $mP$ edges. This corresponds to the last two cases of
Proposition \ref{prop_canonical} (depending on whether $m$ is $R$ or $L$),
therefore proving it.
\end{proof}

\begin{proof}[Proof of Proposition \ref{prop_shortcut}]
We will show how to compute $\tlimit$ algorithmically based on the previous properties.
We initialize the lower bound on $\tlimit$ to 0, and call 
the function {\tt find\_sibling\_limit($v_1$, $v_2$)} for every pair of vertices in $R_1(G)$.
That function compute $v_2-v_1$, and check whether it is smaller than
the length of a path which goes through lower rings; if yes, we update our lower bound on $\tlimit$.
Then, {\tt find\_sibling\_limit} calls itself recursively for every $(w_1, w_2)$ where $w_1$ which is non-rightmost child of $v_1$
and every $w_2$ which is non-rightmost child of $v_2$.

This ensures that every pair of vertices is checked. Of course, this is infinitely many pairs.
However, recursive descent is not necessary if:
\begin{itemize}
\item there is a vertex in the segment $[v_1+1, v_2-1]$ which produces an extra child in every generation.
\item another pair $(v_1, v_2)$ previously considered had the same sequence of types of vertices in $[v_1, v_2]$,
and the same distances from $v_1$ to $v_2-1$ and from $v_1$ to $v_2$ (the results for any pairs of the descendants
of the current pair would be the same as the results for the respective pairs of descendants of the earlier pair).
\end{itemize}

This algorithm is implemented in {\tt regular.cpp} \cite{explorable}. Proposition \ref{prop_shortcut} can be verified
for the given triangulation by running this algorithm or by manual case-by-case analysis.
\end{proof}


\begin{proof}[Proof of Theorem \ref{distalgo}] \rule{0cm}{0cm}

The pseudocode of our algorithm is given in Figure \ref{distalgo_psc}.
It uses five integer variables $a_i,d_i,d$ and four 
vertex variables $l_i$, $r_i$ ($i=1,2$). Variables $a_i$, $d_i$, $l_i$ and $r_i$
are modified only by the function push($i$), which lets us keep the following
invariant: $\dist_0(l_i) = \dist_0(r_i) = d_i$, $l_i = p_L^{a_i}(v_i)$, 
$r_i = p_R^{a_i}(v_i)$.

The lines (\ref{easystart}-\ref{easyend}) deal with the first two cases of the
Proposition \ref{prop_canonical}.

The main loop in lines (\ref{mainloopstart}-\ref{mainloopend}) deals with the last two cases. At all times
$d$ is the currently found upper bound on $\dist(v,w)$. It is easy to check that
the specific shortest path given in Proposition \ref{prop_canonical} will be 
found by our algorithm.

Every iteration of every loop increases $a_1$ or $a_2$, and an iteration can occur
only if $a_1+a_2 < \dist(v,w)$. Therefore, the algorithm runs in time $O(\dist(v,w))$.
An implementation is available (see Appendix \ref{app_impl}, file {\tt segment.cpp}).
\end{proof}

\begin{proof}[Proof of Theorem \ref{powerdistalgo}]
We have not written down the pseudocode nor the proof,
but an implementation is available (see Appendix \ref{app_impl}, file {\tt segment.cpp}).
\end{proof}

\begin{proof}[Proof of Theorem \ref{dhrg_powerlaw}]

{\bf Power law.}

Take a DHRG $V$ with parameters $N$, $\alpha$, $R$, and $T$. Let $X$ be the degree
of a random vertex of $V$. We have to show that $P(X>x) = \Theta(x^{-\beta+1})$.

A random vertex will be
in ring $j$ with probability $\Theta(c^{R-j})$, where $c = e^{-\alpha} < 1$. 

Let $\gamma$ be the growth constant of our RGHT $G$. 
Take two random points $v_1$, $v_2$ from $R_d(G)$. What is the distance between
$v_1$ and $v_2$?

Let $l$ be the distance between $v_1$ and $v_2$ along the cycle, i.e., $v_1+l = v_2$.
From regularity, the distance between $P^k(x_1)$ and $P^k(x_2)$ is then 
$l/\gamma^k + O(1)$. Thus, the algorithm from Theorem \ref{distalgo} will stop after
$s = \log_\gamma(l) + O(1)$ steps, and return $2s + O(1)$.

We could view this as follows: the distance between two random points from $R_d(G)$ is
$2d - 2\min(B,d) + O(1)$, where $B$ has geometric distribution with parameter $1/\gamma$;
intuitively, $B$ corresponds to the length of the common branch of the pathes from $v_0$
to $v_1$ and $v_2$. This formula extends to random points from different rings:
$d(d_1, d_2, B) = d_1 + d_2 - 2\min(B,d_1,d_2) + O(1)$. 

Now, take a DHRG with parameters $N$, $\alpha$, $R$, and $T$. A random vertex will be
in ring $j$ with probability $\Theta(c^{R-j})$, where $c = e^{-\alpha} < 1$. 

We will first consider the step model, where two vertices are connected iff their
distance is $\leq R$. Let $p(d_1, d_2)$ be the probability
that $d(d_1, d_2, X) \leq R$. If $d_1 + d_2 > R$ and $|d_1 + d_2| \leq R$, we must have
$B > (R-d_1-d_2)/2$, thus $p(d_1, d_2) = \Theta(\gamma^{h(d_1+d_2-R)})$ where $h = -1/2$.
The expected degree of a vertex in $R_i(G)$ is:

\begin{eqnarray*}
x(i)/N & = & \sum_{j=0}^{R} \Theta(c^{R-j}) p(d_1, d_2) \\
       & = & \sum_{j=0}^{R-i} \Theta(c^{R-j}) + \sum_{j=R-i+1}^{R} \Theta(c^{R-j}) \Theta(\gamma^{(i+j-R)/-2}) \\
       & = & \Theta(c^R) + \Theta(c^i) + \Theta(c^i) + \Theta(\gamma^{hi}) \\
       & = & \Theta(\max(c,\gamma^h)^i).
\end{eqnarray*}

To take $T$ into account, we simply have to consider that points in distance $R+k$ are connected with
probability $\Theta(q^k)$ for $k \geq 0$ (and $\Theta(1)$ for $k < 0$). Thus, we have replace $p(d_1, d_2)$ with
$\sum_{k=0}^{R-d_1-d_2} q^k \gamma^{h(R-d_1-d_2-k)} = \Theta(\max(\gamma^h,q)^{R-d_1-d_2})$, obtaining
$x(i) = N \Theta(\max(c,\gamma^h,q)^i)$, which is $\Theta(N \gamma^{hi})$ if $\gamma^h>c,q$.
                                                                        
The probability that a random point has degree greater than $x$ is then on the order of probability that
$i < \log_{\gamma^h}(x/N)$, which is $c^{\log_{\gamma^h}(x/N)} = (x/N)^{\log_xi c}$. This proves our
hypothesis with $\log_{\gamma^h} c = -\beta + 1$,
thus $\beta = 1 - \log_{\gamma^h} c = 1 + \alpha / \log \gamma^h = 1 + (\alpha / h\log \gamma) = 1 + 2(\alpha / \log \gamma)$.

We believe that a similar reasoning could be used to theoretically obtain the expected clustering coefficient.
However, the computations are much more complicated
(the three values of $B$ corresponding to each pair of points in the triplet are not independent). \qed

{\bf Algorithm to compute the expected average degree, degree distribution and clustering coefficient of a DHRG.}

For a segment $S$, let $Z^d(S)$ 
be the set of vertices $v$ such that $P^d(v) = S$. For $d>0$, we can compute $|Z^d(S)|$ by considering all the possible 
segments $S'$ such that $P(S') = S$, and summing $Z^{d-1}(S)$ over them. Let the {\it type} of the segment be the
sequence of types of vertices in it (as in the definition of RGHT); $|Z^d(S)|$ depends only on $d$ and the
type of $S$, and there are only finitely many types, so $|Z^d(S)|$ can be computed in $O(d)$ using recursion with
memoization.

Now, let $f(S_1, S_2, d_1, d_2, d)$ be the number of pairs $(v_1 \in Z^{d_1}(S_1), v_2 \in Z^{d_2}(S_2))$ such that 
$\delta(v_1, v_2) = d$. When $S_1$ and $S_2$ are far enough, or $d=0$, we can immediately tell whether $v_1$ and $v_2$ will be in distance $d$;
if yes, the result is $|Z^{d_1}(S_1)| \cdot |Z^{d_2}(S_2)|$, otherwise it is 0. Otherwise we can do recursive computation
in similar way as in the previous paragraph.

By setting $S_1 = S_2 = \{v_0\}$ we obtain the number of pairs of vertices $(v_1, v_2)$ such that $v_1 \in R_{d_1}(G)$,
$v_2 \in R_{d_2}(G)$, and $\delta(v_1,v_0) = d$. Using this information we can easily compute the expected degree of a
random vertex $v \in R_d{G}$, and thus get an approximate expected degree distribution in the DHRG; this is approximate because
the actual expected degree of $v$ depends not only on $d$, but also on the path from $v_0$ to $v$.
An implementation is available (see Appendix \ref{app_impl}, file {\tt dynamic.cpp}).

The clustering coefficient can be computed in a similar way, but we have to consider triplets of points. 
\end{proof} 

\section{Implementation}\label{app_impl}

The source code, data, and experimental results are
available at the following address:

\begin{center}
\url{http://www.mimuw.edu.pl/~erykk/dhrg/dhrg-v5.tgz}

md5sum: {\tt cf73cd3045f37dfbf73d16b143e4eaa6}
\end{center}
 
Here v5 represents the version at the time of this submission.

The following elements are included:

\begin{itemize}
\item {\tt rogueviz} and {\tt src} -- 
implementation of the algorithms and data structures from this paper. This 
builds on RogueViz, which is a hyperbolic visualization/analysis engine based on HyperRogue \cite{hyperrogue}.
RogueViz implements:
\begin{itemize}
\item regular generation of $\scht{q}$ and $\schq{q}$ grids (heptagon.cpp)
\item Goldberg-Coxeter construction (goldberg.cpp)
\item computing types of vertices for regular generation, the function $c:T \ra T^*$,
computing the growth factor $\gamma$ based on $T$ and $c$, and computing the 
distance based on Algorithm \ref{distalgo} (expansion.cpp)
\item mapping tesselation vertices to the hyperbolic space and vice versa
\item visualization engine
\end{itemize}             
src implements algorithms discussed in this paper:
\begin{itemize}
\item an algorithm to compute $\tlimit$ (regular.cpp)
\item RGHT structure as used in this paper (mycell.cpp)
\item segments, and implementation of the algorithm from Theorem \ref{distalgo} (segment.cpp)
\item log-likelihood analysis and local search to improve embedding (loglik.cpp, embedder.cpp)
\item distance algorithm mentioned in the proof of Theorem \ref{dhrg_powerlaw} (dynamic.cpp)
\item a function to test Conjecture \ref{distconj} (gridmapping.cpp)
\end{itemize}
\item {\tt embedded-graphs} -- embedded graphs. This includes the FIT 2017 coauthorship network 
(a very small network for quick testing and visualization), GitHub following networks from 2009 and 2011, 
and some of the networks that the BFKL embedder was benchmarked on (Facebook, Amazon, Slashdot).
\item {\tt results} -- detailed results of the local search on various graphs, and of 
gridmapping.cpp and dynamic.cpp.
\item {\tt oldresults} -- results from an older version which have not yet 
been recomputed with the newer version. Note that the newer version uses much less memory.
\item {\tt web} -- a copy of the browser-based interactive visualization.
\end{itemize}

Look at the Makefile to see how to obtain various targets.
Run {\tt make visualize} to visualize the local search process on the FIT network.
Press WASD or left-click to move around, / to display the statistics, 
display log-likelihood, manually move the vertices to see the effect on the log-likelihood, and
execute iterations of the algorithm.

\begin{figure}[ht]
\begin{center}
\includegraphics[width=.26\textwidth]{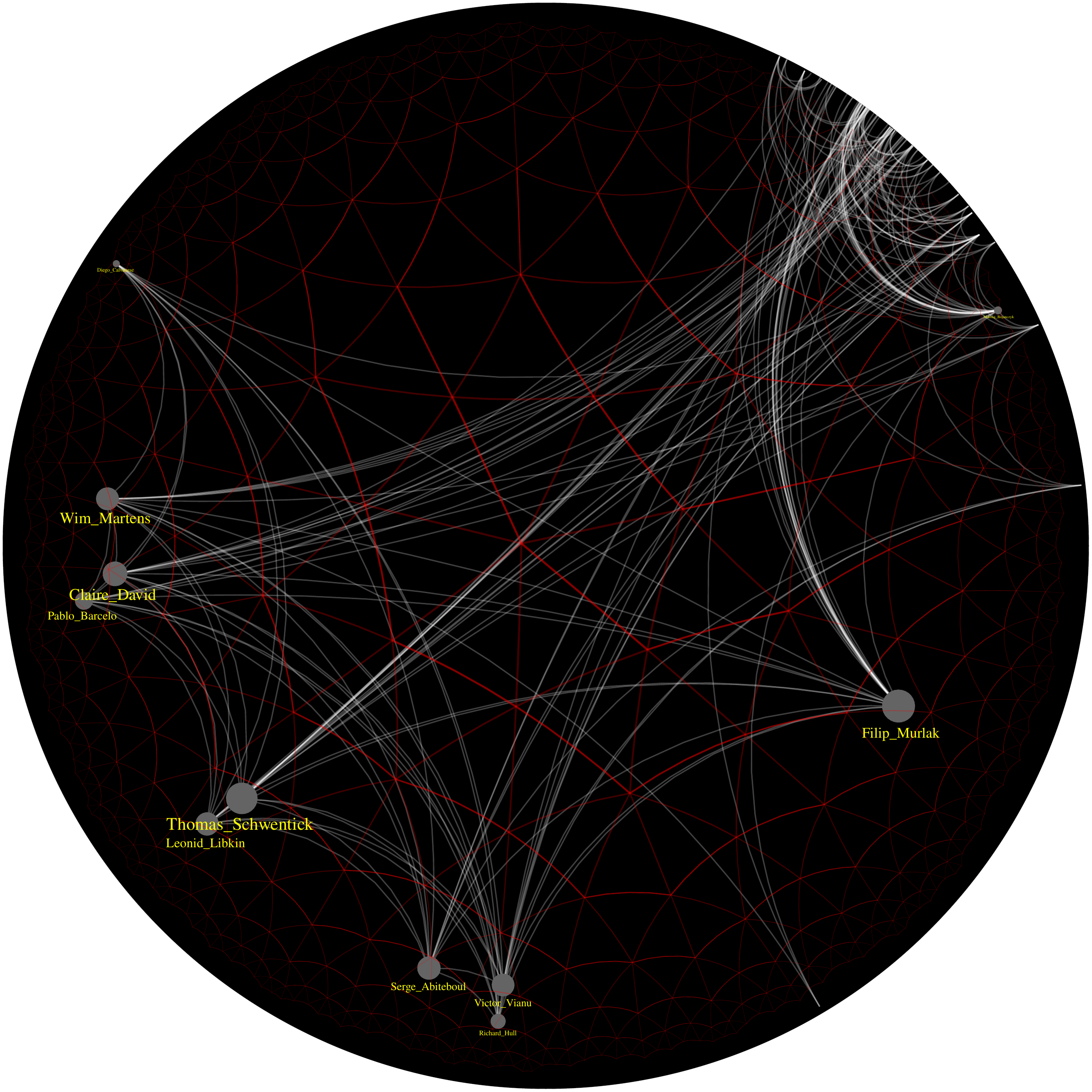}
\includegraphics[width=.26\textwidth]{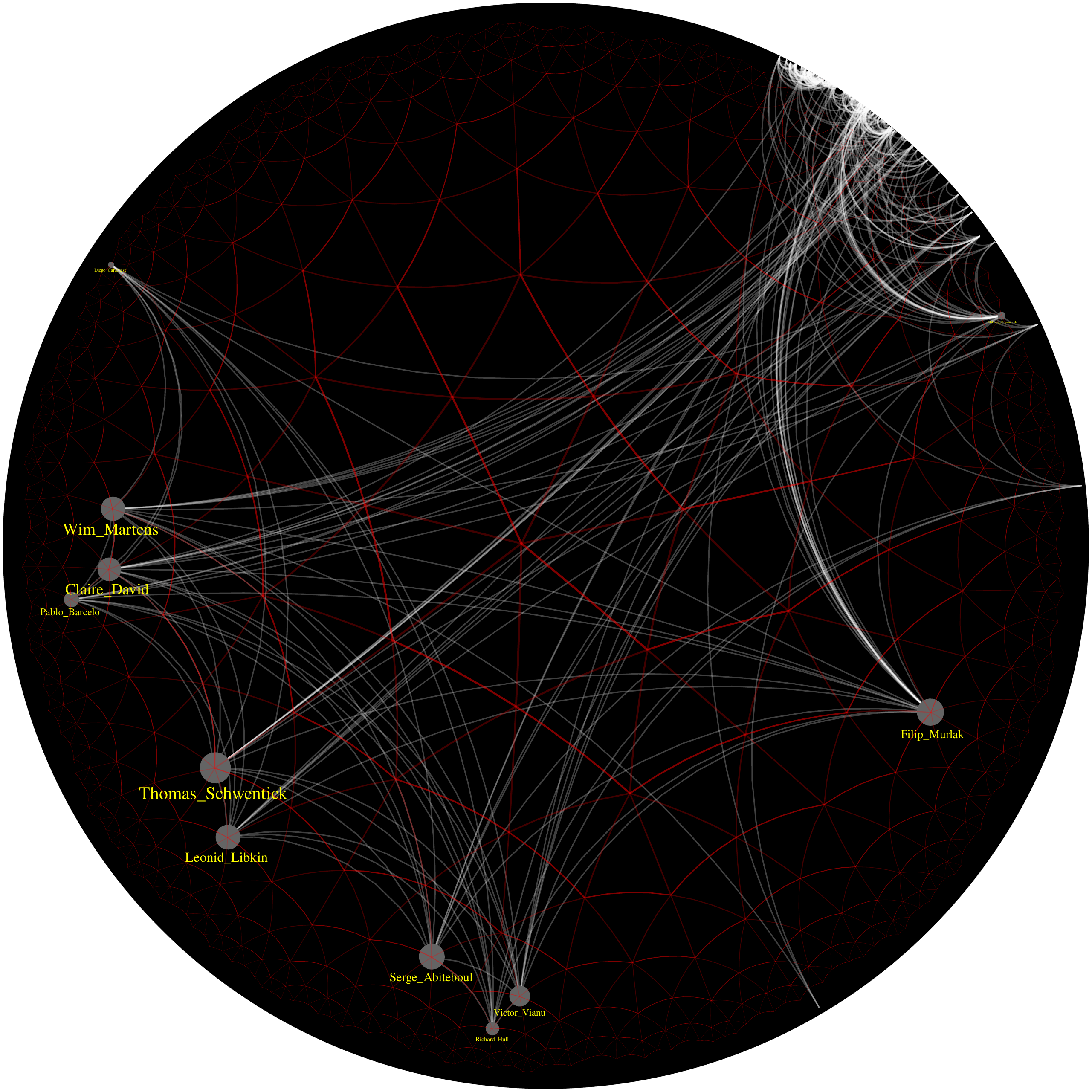}
\includegraphics[width=.26\textwidth]{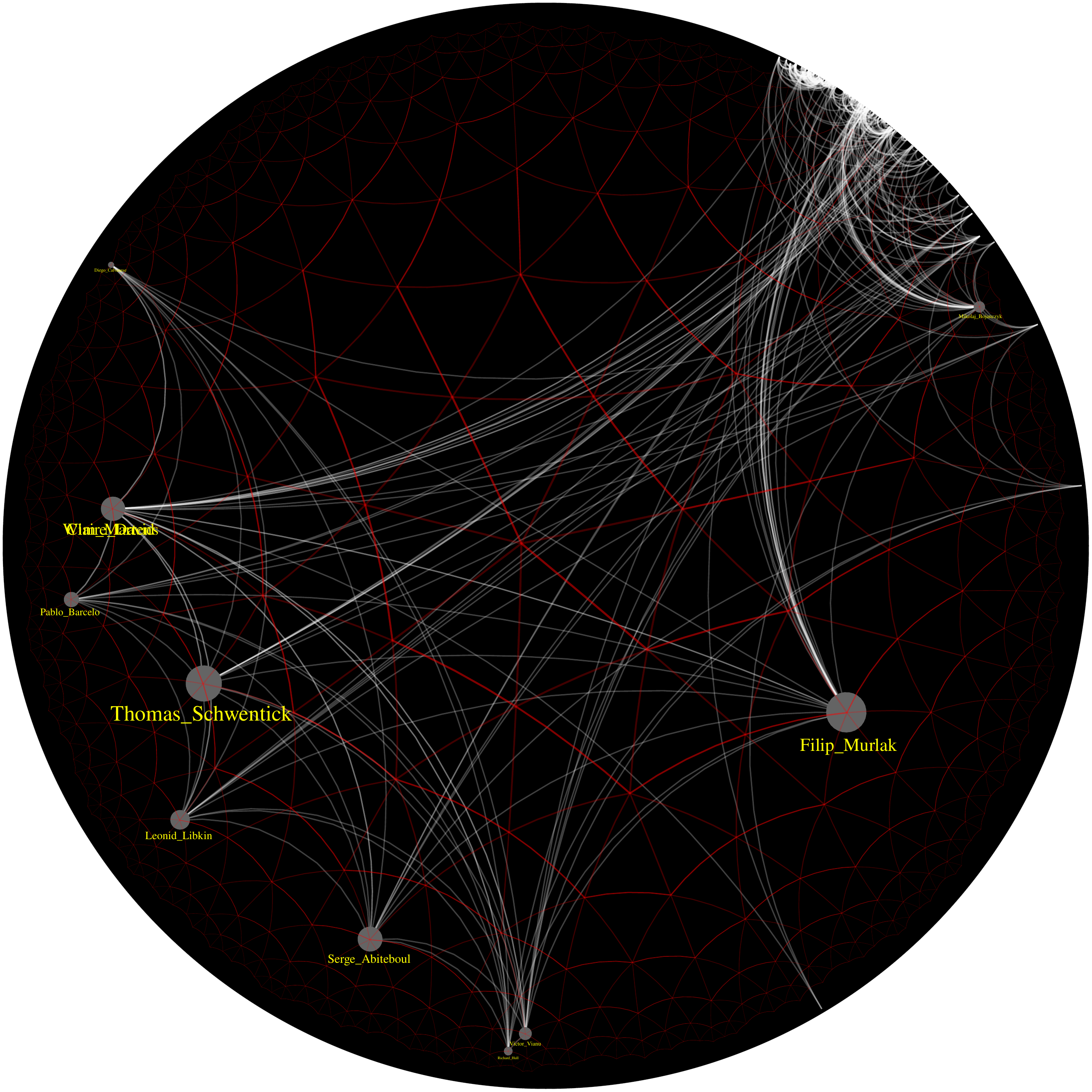}
\end{center}
\caption{\label{figvis}
Visualization of the FIT 2017 coauthorship network. 
The root vertex is far away at 1 o'clock,
the picture is centered on a specific cluster of
authors.
These authors collaborated on many papers;
many of them have also collaborated with other
authors in the network.
From left to right: 
the network embedded with the BFKL embedder;
vertices of the network are moved to the vertices of $\ghoo$;
one step of the local search.}
\end{figure}

Some experiments have not been mentioned in the main paper.
The local search can optimize one of the following measures: logistic log-likelihood based on the optimal
values of $R$ and $T$; optimal log-likelihood where edge probability is given separately for each distance;
monotonic optimal log-likelihood where the probability function has to be decreasing with larger distances;
total entropy obtained by summing the optimal log-likelihood of edge and vertex placement. Non-monotonic
optimal log-likelihood tends to scapegoat a fixed small distance (say, 3) and put all the pairs of close vertices
which are not actually connected at that distance; monotonic optimal does not have this problem. Entropy
minimization could be potentially used as a compression method; a quite good compression (46\%) is obtained for
the Facebook graph, though bigger graphs do not compress that well. An alternative non-local method of 
improving embedding is implemented, where vertices can immediately move to good locations far away (we start in the
center and move in the most promising direction); this improves the log-likelihood somewhat.

The current version uses a significant amount of RAM (2.4 GB for 6 iterations the followers-2009 network which has 74946 vertices,
on $\ghoo$; on $\ghoz$ it uses 1.4 GB, and on $\gooz$ it uses 1.2 GB). 
It should be possible to improve this by better memory management (currently vertices and segments which are
no longer used or just temporarily created are not freed), or possibly path compression. 

\end{document}